\documentclass{article}

\usepackage{amssymb,amsmath}
\usepackage{color}
\usepackage{algpseudocode}
\usepackage[ruled,vlined]{algorithm2e}
\usepackage{amsthm}
\usepackage{graphicx}   
\usepackage{cite}
\usepackage{textcomp}
\usepackage{hyperref}
\hypersetup{hidelinks=true}
\usepackage[margin=1in]{geometry}

\DeclareMathOperator*{\argmin}{argmin}
\DeclareMathOperator*{\argmax}{argmax}

\newcommand{\rev}[1]{{\color{black} {#1}}}

\newcommand{\R}{\mathbb{R}}
\newcommand{\e}{\varepsilon}
\newcommand{\E}{\mathbb{E}}
\newcommand{\F}{\text{F}}
\newcommand{\Lh}{L_\text{high}}
\newcommand{\Ll}{L_\text{low}}
\newcommand{\hx}{\hat{x}}
\newcommand{\sigbm}{\sigma^{\text{BM3D}}}
\newcommand{\s}{\mathfrak{p}}

\newtheorem{theorem}{Theorem}[section]
\newtheorem{definition}[theorem]{Definition}

\newtheorem{assump}{Assumption}

\def\BibTeX{{\rm B\kern-.05em{\sc i\kern-.025em b}\kern-.08em
    T\kern-.1667em\lower.7ex\hbox{E}\kern-.125emX}}
\AtBeginDocument{\definecolor{tmlcncolor}{cmyk}{0.93,0.59,0.15,0.02}\definecolor{NavyBlue}{RGB}{0,86,125}}

\begin{document}


\title{Denoiser-based projections for 2D super-resolution MRA}

\author{Jonathan Shani, Tom Tirer, Raja Giryes, and Tamir Bendory}



\maketitle

\begin{abstract}
	We study the 2D super-resolution multi-reference alignment (SR-MRA) problem: estimating an image from its down-sampled, circularly translated, and noisy copies.
	The SR-MRA problem serves as a mathematical abstraction of the structure determination problem for biological molecules.
	Since the SR-MRA problem is ill-posed without prior knowledge, accurate image estimation relies on designing priors that describe the statistics of the images of interest. 
	In this work, we build on recent advances in image processing and harness the power of denoisers as priors for images. 
	To estimate an image, we propose utilizing denoisers as projections and using them within two computational frameworks that we propose: projected expectation-maximization and projected method of moments. We provide an efficient GPU implementation and demonstrate the effectiveness of these algorithms through extensive numerical experiments on a wide range of parameters and images.
\end{abstract}

\section{Introduction}
2D super-resolution multi-reference alignment (SR-MRA) entails estimating an image $x \in \mathbb{R}^{\Lh\times{\Lh}}$ from its $N$ circularly-translated, down-sampled, noisy copies:
\begin{equation} \label{eq:model_1}
	y_i = PR_{s_i}x+\e_i, \quad i= 1,\ldots,N,
\end{equation}
where $R_{s}$ denotes a 2D circular translation,
$P$ denotes a down-sampling operator that collects $\Ll\times \Ll$ equally-spaced samples of $R_sx$, and
$\e_i \in \mathbb{R}^{\Ll\times{\Ll}}$ is a noise matrix whose entries are drawn i.i.d.\ from $\mathcal{N}(0,\sigma^2)$.
The 2D translation~$s$ is composed of a horizontal translation $s^1$ and a vertical translation $s^2$, which are drawn i.i.d.\ from unknown distributions, $\rho_1$ and $\rho_2$, respectively. 
Explicitly, each observation $y_i\in \mathbb{R}^{\Ll\times{\Ll}}$ takes the form
$y_{i}[n_1,n_2] =
 x\left[n_1K-s_i^1, n_2K-s_i^2\right]  +\e_i[n_1,n_2],$
where $n_1,n_2=0,\ldots,\Ll-1$, the shifts should be considered modulo $\Lh$, and $K := \frac{\Lh}{\Ll}$ is assumed to be an integer. 
Our goal is to estimate $x$ (the high resolution image) from~$N$ low-resolution observations $y_1,\ldots,y_N,$ when the translations $s_1,\ldots,s_N$ are unknown. 

The SR-MRA model, first studied for 1-D signals~\cite{bendory2020super}, is a special case of the multi-reference alignment (MRA) model: The problem of estimating a signal from its noisy copies, each acted upon by a random element of some group; see for example~\cite{bandeira2015non,bandeira2017estimation,bendory2017bispectrum,abbe2018multireference,perry2019sample,bendory2022dihedral,zhao2022adaptive}.
The MRA model is mainly motivated by the single-particle cryo-electron microscopy (cryo-EM) technology: an increasingly popular technique to construct 3-D molecular structures~\cite{bendory2020single}. 
In Section~\ref{sec:conclusion}, we introduce the mathematical model of cryo-EM in detail 
and discuss how the techniques proposed in this paper have the potential to make a significant impact on the molecular reconstruction problem of cryo-EM.
The MRA model is also motivated by a variety of applications in biology~\cite{park2011stochastic}, robotics~\cite{rosen2020certifiably}, radar~\cite{zwart2003fast}, and image processing~\cite{bendory2021compactification}.

The two leading computational techniques for solving MRA problems are the method of moments (MoM) and expectation-maximization (EM) \cite{bendory2017bispectrum, abbe2018multireference}.
MoM is a classical parameter estimation technique, aiming to recover the parameters of interest
from the observed moments. 
MoM requires a single pass over the observations, making it an efficient technique for large data sets (large $N$) but is not statistically efficient.
The EM algorithm aims to maximize the likelihood function (or the posterior distribution in a Bayesian framework)~\cite{dempster1977maximum}.
As Theorem~\ref{thm:uniqueness} shows,   
it is impossible to uniquely identify the image $x$ only from the SR-MRA observations~\eqref{eq:model_1}
as many different images result in the same likelihood function.
Thus, to estimate the image accurately, we need to incorporate a prior. 
In this work, we harness a recent line of works that use denoisers as priors and show how to incorporate them into the MoM and EM in the context of SR-MRA.

Utilizing existing state-of-the-art denoisers was proposed as part of the Plug-and-Play technique~\cite{venkatakrishnan2013plug} and was proven  highly effective for many imaging inverse problems~\cite{RED, PnP_ADMM,Heide2014FlexISP,zhang2017learning, tirer2018image,tirer2019super,DPIR,PnP19Ryu,Rond16Poisson,kamilov2017plug,
wei2020tuningfree,Meinhardt17Learning,Buzzard18PnP,Sun19Online,
Reehorst19RED,Ryu2019PlugandPlayMP,Sun2019Block,Lai2023Nabla}.   
Essentially, the underlying idea is exploiting the impressive capabilities of existing denoising algorithms to replace explicit, traditional priors. 
Indeed, these methods are especially useful for natural images, whose statistics are too complicated to be described explicitly, but for which excellent denoisers have been devised. 
A natural question is why not just use the standard plug-and-play framework for SR-MRA instead of developing a dedicated strategy for EM and MoM as we do in this work. The simple answer is that, despite intensive efforts,
we could not combine the plug-and-play approach  
with MoM and EM in a stable way. Therefore, we found that it is necessary to develop a dedicated strategy as we do in this work.

Our proposed computational framework, presented in Section~\ref{sec:projected_alg}, uses denoisers as projection operators. Specifically, every few iterations of either MoM or EM, we apply a denoiser to the current estimate. 
This stage can be interpreted as projecting the image estimate onto the implicit space spanned by the denoiser.
Section~\ref{sec:Convergence} provides  convergence analysis, and 
Section~\ref{sec:experiments} demonstrates the effectiveness of our method by extensive numerical experiments on images. 

The contribution of this paper may be summarized as follows: (i) We propose a novel framework for the SR-MRA problem that relies on using denoisers as priors throughout the optimization; (ii) we show that the method is generic by applying it to two popular methods for SR-MRA, namely MoM and EM; (iii) we provide an initial analysis for the convergence of our proposed algorithms; (iv) we open-source the code of our approach that efficiently implements our strategy on a GPU for getting efficient parallel processing. 

\section{Background}
\label{sec:background}

Before introducing our framework, we elaborate on four pillars of this work: the non-uniqueness of the SR-MRA model~\eqref{eq:model_1}, MoM, EM, and denoiser-based priors.
Hereafter, we let $\rho:=[\rho_{1}, \rho_{2}]^T$, and 
define the set of sought parameters as $\theta:=(x,\rho)$. 
With a slight abuse of notation, we treat the image $x$ and a measurement $y_i$ as  vectors in $\R^{\Lh^2}$ and $\R^{\Ll^2}$, respectively. 
We also define $\Theta=\R^{\Lh^2}\times\Delta_{\Lh}\times \Delta_{\Lh}$, where $\Delta_{\Lh}$ is the $\Lh$-dimensional simplex, 
so that $\theta\in\Theta$. 

\noindent {\bf SR-MRA observations do not identify the image $x$ uniquely.} 
The following theorem states that the likelihood function of~\eqref{eq:model_1} does not determine the image $x$ uniquely. This, in turn, implies that a prior is necessary for accurate estimation of the sought image. 

\begin{theorem} \label{thm:uniqueness}
	The likelihood function $p(y_1,\ldots,y_N|x,\rho)$ does not determine uniquely the sought parameters $x$ and $\rho$. 
\end{theorem}
\noindent \emph{Proof.} We follow the proof of~\cite[Theorem 3.1]{bendory2020super}.	
	Recall that $y = PR_{s}x+\varepsilon$, and that $K:=\frac{\Lh}{\Ll} $ is an integer. Let us define a set of $K^2$ sub-images, indexed by $n_1,n_2=0,\ldots,K-1$:
	\begin{equation*}
		x_{n_1,n_2}[\ell_1,\ell_2] := x[n_1+\ell_1K ,n_2+\ell_2K], 
	\end{equation*}
	for $\ell_1,\ell_2=0,\ldots,\Ll-1.$
	Then, the SR-MRA model reads 
	\begin{equation}
		y = R_{t}x_{n_1,n_2} + \varepsilon,    
	\end{equation}
	where $R_t$ is a translation over the low-resolution grid of size $\Ll\times\Ll$. 
	We  denote the distribution of choosing  the sub-image $x_{n_1,n_2}$ and  translating it by $t$ as $\rho[n_1,n_2,t]$.
	Thus, the likelihood function of~\eqref{eq:model_1}, for a single observation $y$, can be written, up to a constant, as
	\begin{equation*}
		p(y; x,\rho) =  \sum_{n_1,n_2=0}^{K-1}\sum_{t}\rho[n_1,n_2,t]e^{-\frac{1}{2\sigma^2}\|y -R_{t}x_{n_1,n_2}\|_2^2}.
	\end{equation*}
	Note that the likelihood function is invariant under any 
	permutation of the sub-images (overall $K^2!$ permutations), and under any translation of each of them (overall $(\Ll^2)^{K^2}$ possible translations). 
	This implies that there are $K^2!(\Ll^2)^{K^2}$ different images with the same likelihood function.
\qed

Figure~\ref{fig:ref_images_MoM} compares image recovery  
using projected MoM and projected EM, which are our proposed algorithms (see Section~\ref{sec:projected_alg}). Specifically, the first and third rows show recoveries using Algorithm~\ref{proj_MoM} and Algorithm~\ref{proj_EM}, respectively, while the second and fourth rows show the outputs of the standard MoM and EM, with no prior. 
Evidently, in light of Theorem~\ref{thm:uniqueness}, the latter results are of low quality.
 
\begin{figure}[t]
    \centering    \includegraphics[width=0.5\textwidth]{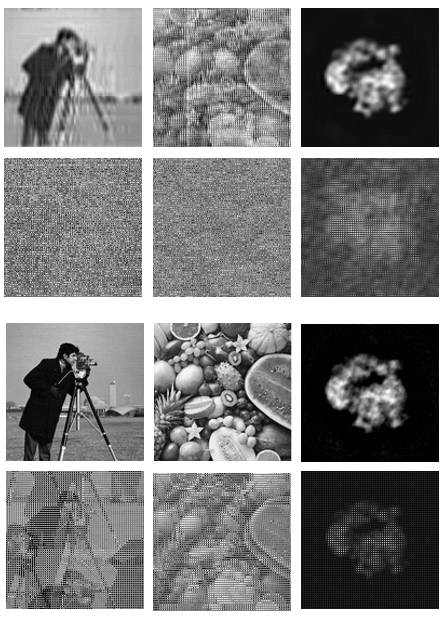}
    \caption{
	The first and third  rows present  recoveries using 
	projected MoM and projected EM.
	Second and fourth rows present recoveries using MoM and EM, without projection.  
    As Theorem~\ref{thm:uniqueness} indicates, the projection is vital for accurate recovery. 	
    	Images are of size $\Lh=128$ with a down-sampling factor of $K=2$. MoM used the exact moments, and EM was applied with noise level $\sigma=1/8$ and number of observations $N=10^4$.
    	In the projected versions, we set $\F=5$ (see Section~\ref{sec:projected_alg}).
    }
    \label{fig:ref_images_MoM}
\end{figure}

\noindent {\bf The method of moments (MoM)} aims to estimate the parameters of interest 
from the observable moments.
Recall that the  $d$-th observable moment is given by 
\begin{equation} \label{eq:empiric_MOM}
	\widehat{M_d} := \frac{1}{N} \sum_{i=1}^N y_i^{\bigotimes d}, 
\end{equation}
where ${y}^{\bigotimes d}$ is a tensor with $(\Ll^2)^{d}$ entries, and the entry indexed by $\ell=(\ell_1,\ldots,\ell_d)$ is given by $\prod_{j=1}^d y[\ell_j]$.
Computing the observable moments requires a single pass over the observations and results in a concise summary of the data. 
By the law of large numbers, for sufficiently large $N$, we have 
\begin{equation} \label{eq:moments_convergence}
    \widehat{M_d}\approx \E y^{\bigotimes d}:=M_d(\theta),
\end{equation}
where the expectation is taken against the distribution of the translations and the noise.  
Finding the optimal parameters~$\theta$ 
that fit the observable moments 
 is usually performed by minimizing a least squares objective, 
\begin{equation} \label{eq:LS}
\min_{\theta\in\Theta} \sum_{d=1}^D \lambda_d \| \hat{M_d} - M_d(\theta) \|^2,
\end{equation}
where $\lambda_1,\ldots,\lambda_D$  are some predefined weights.

Model~\eqref{eq:model_1} was studied in~\cite{abbe2018multireference} when $P$ is the identity operator (no down-sampling)\footnote{The results of~\cite{abbe2018multireference} concern 1-D, but they can be readily extended to 2-D.}. 
In particular, it was shown that if the discrete Fourier transform of the signal is non-vanishing, and $\rho$ is almost any non-uniform distribution, then the signal and distribution are determined uniquely, up to an unavoidable translation symmetry, from the second moment of the observations.
While this is not necessarily true when $P$ is a down-sampling operator, we keep using the first two moments in this work (namely, $D=2$).

Specifically, the first two moments of the SR-MRA model~\eqref{eq:model_1} are given by~\cite{abas2021generalized}:
\begin{eqnarray}
    M_1 &=& PC_{x}\rho, \\
    M_2 &=& PC_{x}D_{\rho}C_{x}^TP^T + \sigma^2PP^T,
\end{eqnarray}
where $C_{x} \in \mathbb{R}^{\Lh^2 \times \Lh^2}$ is a block circulant with circulant blocks (BCCB)  matrix of the form:
\begin{equation}
C_{{x}}  = \begin{pmatrix}
c_0 & c_{\Lh-1} & \ldots & c_2 & c_1\\
c_1 & c_0 & c_{\Lh-1} & \ldots & c_2\\
\ldots & \ldots & \ldots & \ldots & \ldots\\
\ldots & \ldots & \ldots & \ldots & \ldots \\
c_{\Lh-1} & c_{\Lh-2} & \ldots & c_1 & c_0
\end{pmatrix},
\end{equation}
where the block $c_i\in \mathbb{R}^{\Lh\times \Lh}$ represents a circulant matrix of the $i$-th column of the image $x$.
Each column of $C_x$ contains a copy of the image after 2D translation and vectorization.
The matrix $D_{\rho} \in \mathbb{R}^{\Lh^2\times \Lh^2}$ is a diagonal matrix, whose diagonal is a vectorization of the matrix $\rho_{1}\rho_{2}^T$.
Therefore, the least squares objective~\eqref{eq:LS} can be explicitly written as 
\begin{eqnarray} \label{eq:LS_SR}
 \nonumber && \hspace{-0.3in} \min_{ \theta\in\Theta}  || PC_{{x}}D_{\rho}C_{{x}}^TP^T + \sigma^2PP^T -\hat{M_2}||_{\F}^2 \nonumber \\ && ~~~~~~~~~~~~~~~~~~ +\lambda||PC_{\rho}{x}-\hat{M_1}||_{2}^2.
\end{eqnarray}
In this paper, we set $\lambda=\frac{1}{\Lh^2(1+\sigma^2)}$, as suggested in~\cite{abbe2018multireference}.

\noindent {\bf Expectation-maximization (EM).} 
The log-likelihood of~\eqref{eq:model_1} is given by
\begin{equation}\label{eq:likelihood_EM}
    \log p(y_1,\ldots,y_n; \theta) =  \sum_{i=1}^N\log \sum_{s^1,s
    ^2=1}^{\Lh}\rho[s]e^{-\frac{1}{2\sigma^2}\|y_i -PR_{s}x \|_2^2}.
\end{equation}
A popular way to maximize the likelihood function is using the EM algorithm~\cite{dempster1977maximum}. 
EM can also be used to maximize the posterior distribution 
when an analytical prior is available. 

EM is an iterative algorithm, where each iteration consists of two steps. The first, called the E-step, computes the expected value of the log-likelihood function with respect to the translations, given the current estimate of $\theta:=(x,\rho)$.
In the $(t+1)$-th iteration, it reads 
\begin{equation} \label{eq:Q}
	\begin{split}
&Q(x,\rho | x_{t},\rho_t) = \\& \sum_{i=1}^{N}\sum_{s^1,s^2=1}^{\Lh}w_t^{i,s}    \bigg(\log{\rho[s]}-\frac{1}{2\sigma^2}||y_i-PR_{s}x||^2_\F\bigg),
	\end{split}
\end{equation}
where
\begin{equation}\label{eq:weights_EM}
    w_t^{i,s} = C_t^{i}{\rho_t}[s]e^{-\frac{1}{2\sigma^2}||y_i-PR_{s}x||^2_\F},
\end{equation}
and $C_t^{i}$ is a normalization factor so that $\sum_{s}w_t^{i,s}=1$. 
The next step, called the M-step, maximizes~\eqref{eq:Q} with respect to $x$ and $\rho$:
\begin{eqnarray}
(x_{t+1},\rho_{t+1}) = \argmax_{ x, \rho} Q(x,\rho | x_t,\rho_t).
\end{eqnarray} 
In our case, the maximum is attained by solving the linear system of equations:
\begin{equation} \label{eq:x_EM}
      Ax_{t+1} = b,
\end{equation}
where 
\begin{align*}
A &= \sum_{i=1}^{N}\sum_{s}w_t^{i,s}R_{s}^{-1}P^{-1}PR_{s},  \nonumber \\
b &= \sum_{i=1}^{N}\sum_{s}w_t^{i,s}R_{s}^{-1}P^{-1}y_i, \nonumber \\
\end{align*}
and
\begin{eqnarray}
\label{eq:rho_EM}
    \rho_{t+1}[s] &= \frac{ \sum_{i=1}^{N}w_k^{i,s}}{\sum_{i=1}^{N}\sum_{s'}w_k^{i,s'}}.
\end{eqnarray}

The EM algorithm iterates between the E and the M steps. 
While  each EM iteration does not decrease the likelihood, 
for the SR-MRA model, the likelihood function is not convex, and thus 
the EM algorithm is not guaranteed to achieve its maximum.
A common practice to improve estimation accuracy is either to initialize the EM iterations in the vicinity of the solution (if such an initial estimate is available) or to run it multiple times from different initializations, and choose the estimate corresponding to the maximal likelihood value. 
As we show in Section~\ref{sec:experiments}, our projected EM algorithm (described in the next section) provides consistent results even with a single initialization.

\noindent {\bf Denoiser-based priors.}
Using denoisers as priors was proposed in the Plug-and-Play approach \cite{venkatakrishnan2013plug}. 
It aims to maximize a posterior distribution $p(\theta|y_1,\ldots,y_n)$, which is equivalent to the optimization problem:
\begin{equation}
	 \argmin_{\theta\in\Theta} l(\theta)  + \s(\theta),
\end{equation}
where $l(\theta):=-\log{p(y_1,\ldots,y_n|\theta)}$, $\s(\theta):=-\log p(\theta)$, $p(\theta)$ is a prior on $\theta$, and 
recall that $\theta=(x,\rho)$. 
This work only considers a prior on the image $x$, but a prior on the distribution might be learned as well~\cite{janco2021accelerated}.

Plug-and-Play proceeds by variable splitting, 
\begin{equation}
\label{eq:opt_prob_constrained}
   \argmin_{\theta,v\in\Theta} l(\theta) + \beta \s(v) \quad 
    \text{subject to} \quad  \theta=v,
\end{equation}
where $\beta$ is a parameter that is manually tuned. 
The problem in \eqref{eq:opt_prob_constrained} can be solved using ADMM by iteratively applying the following three steps:
\begin{eqnarray}
\label{eq:admm_steps}
   \theta^{k+1} &=& \argmin_{\theta\in\Theta} l(\theta) + \frac{\lambda}{2} ||\theta - (v^k-u^k)||_2^2, \\ \nonumber
   v^{k+1} &=& \argmin_{v\in\Theta} \frac{\lambda}{2\beta} ||\theta^{k+1} + u^{k} - v||_2^2 + \s(v),
   \\ \nonumber
   u^{k+1} &=& u^{k} + (\theta^{k+1}-v^{k+1}),
\end{eqnarray}
where $\lambda$ is the ADMM penalty parameter.
The main insight here is that the second step aims to solve a Gaussian denoising problem, with a prior $\s(v)$ and 
standard deviation of $\sqrt{\frac{\beta}{\lambda}}$.
Instead of solving this problem 
for an analytical explicit prior, one can apply off-the-shelf denoisers, such as BM3D or neural nets.

While Plug-and-Play was used in many tasks, it requires tuning the parameters $\lambda, \beta$ and the number of iterations for the first step; these parameters significantly impact its performance.
In our setup, we did not find a set of parameters that yields accurate estimates for the SR-MRA problem in a wide range of scenarios. 
As an alternative, we propose 
a projection-based approach for EM and MoM. 

\section{Denoiser-based projected algorithms} \label{sec:projected_alg}

This section introduces the main contribution of this work: incorporating a projection-based denoiser into the MoM and EM, and their implementation for SR-MRA. 
Importantly, there is no agreed way of including a prior into MoM, nor a standard way of incorporating natural images prior to EM.
Projected MoM algorithm is presented in Algorithm~\ref{proj_MoM}, and projected EM in Algorithm~\ref{proj_EM}.

The main idea is rather simple: every few iterations of either MoM or EM, we apply a predefined denoiser that acts on the current estimate of the image. 
Specifically, for MoM, we apply the denoiser every few steps of BFGS (with line-search) that minimizes the least squares objective~\eqref{eq:LS_SR}; 
BFGS can be replaced by alternative gradient-based methods.
For EM, we apply the denoiser every few EM steps. 

As a denoiser, we chose to work with the celebrated BM3D denoiser~\cite{dabov2007image} because of its simplicity and effectiveness for denoising natural images. 
However, replacing BM3D with alternative denoisers is straightforward. 
This might be especially relevant when dealing with specific applications, where a data-driven deep neural net may take the role of the denoiser. In Section~\ref{sec:conclusion}, we discuss such potential applications for cryo-EM data processing. 

Our algorithm requires only two parameters to be adjusted. The first parameter is the noise level of the denoiser, denoted hereafter by $\sigbm$. 
A low noise level increases the weight of the data (either likelihood or observed moments), while increasing the noise level strengthens the effect of the denoiser.
In particular, we found that starting from a high noise level and gradually decreasing it with iterations leads to consistent and accurate estimations. 
Specifically, in the experiments of Section~\ref{sec:experiments}, we set the noise level to 
\begin{eqnarray} \label{eq:decaying_function}
    \sigbm_t = 2^{-1/10}{\sigbm_{t-1}}, \quad \sigbm_1=1,
\end{eqnarray}
where $\sigbm_t $ is the noise level of the $t$-th application of BM3D.  
 
The second parameter, which we denote by $F$, determines the number of EM iterations or BFGS iterations (or other gradient-based algorithms) for MoM between consecutive applications of the denoiser.
Small or large $F$ may put too much or too little weight on the prior (rather than on the data).
Empirically, projected MoM is quite sensitive to this parameter and requires carefully tuning $F$, as presented in Section~\ref{sec:experiments},  whereas EM is less sensitive.
In contrast to Plug-and-Play, which was very sensitive to parameter tuning in our tests and did not converge well, we found it to be quite easy to tune the parameters $F$ and $\sigma$ and get consistent results for a wide range of images and parameters. 
The simplicity of parameter tuning, perhaps, stems from the fact that each parameter affects only one term ($F$ the fidelity, and $\sigbm$ the prior), 
whereas tuning the $\lambda$ parameter in the Plug-and-Play approach~$\eqref{eq:admm_steps}$ affects both terms. 

\begin{algorithm}[t]
    \textbf{Input}: Measurements $y_1,\ldots,y_N$, denoiser $\mathcal{D}$, and the parameter $F$\\ 
         \textbf{Output}: An estimate of the target image and  distribution

    \begin{enumerate}
        \item     Compute the empirical moments $\hat M_1, \hat M_2$ according to~\eqref{eq:empiric_MOM} 
        \item Until a stopping criterion is met:
        \begin{enumerate}
           \item Run $F$ BFGS (or alternative gradient-based algorithm) steps to minimize~\eqref{eq:LS_SR}
           \item Apply the denoiser $x \gets \mathcal{D}(x, \sigbm$) 
         \item Update $\sigbm$ according to~\eqref{eq:decaying_function}
        \end{enumerate}  
    \end{enumerate}
    \caption{Projected MoM}\label{proj_MoM}
\end{algorithm}

\begin{algorithm}[t]
    \textbf{Input}: Measurements $y_1,\ldots,y_N$, denoiser $\mathcal{D}$, and the parameter $F$
    
         \textbf{Output}: An estimate of the target image and  distribution

    Until a stopping criterion is met:
    
           \begin{enumerate}
           \item Run $F$ EM steps according to~\eqref{eq:weights_EM},~\eqref{eq:x_EM}, and ~\eqref{eq:rho_EM}.
           \item Apply the denoiser  $x \gets \mathcal{D}(x, \sigbm)$
         \item Update $\sigbm$  according to~\eqref{eq:decaying_function}
        \end{enumerate}  
\caption{Projected EM}
    \label{proj_EM}
\end{algorithm}

\section{Convergence analysis}
\label{sec:Convergence}

We turn to analyze the convergence of the proposed algorithms.
We start with a couple of definitions and assumptions that will be used throughout the analysis. 
Then, we provide convergence analysis for projected MoM in a somewhat modified scheme, where the (nonconvex) MoM objective is minimized by a gradient descent step rather than BFGS iterations (which are essentially preconditioned gradient steps). 
Lastly, to tackle frameworks that are not based on any gradient method, we present a more general operator-based analysis. While some assumptions in this analysis are not guaranteed for EM, it provides mathematical reasoning for the convergence of frameworks that are composed of two ``black-box" operators that generate convergent sequences, each on its own, which is relevant to EM. 

\begin{definition}
We say that the map $\mathcal{M}:\mathcal{X} \to \mathcal{X}$ is $C$-Lipschitz in $\mathcal{X}$ if for all $x,z \in \mathcal{X}$ we have 
$$
\|\mathcal{M}(x) - \mathcal{M}(z)\|_2 \leq C \| x-z\|_2.
$$
In particular, if $C=1$ then $\mathcal{M}$ is nonexpansive in $\mathcal{X}$ and if $C<1$ then $\mathcal{M}$ is a contraction in $\mathcal{X}$.
\end{definition}

\begin{definition}
\label{def:prox}
The proximal operator associated with a function $\s(\cdot)$ is defined by
$$
\mathrm{Prox}_{\s}(x) := \argmin_z \frac{1}{2}\|x-z\|_2^2 + \s(z).
$$
\end{definition}
Observe the tight connection between the proximal operator and Gaussian denoising.
Specifically, Gaussian denoiser $\mathcal{D}(\cdot;\sigma)$ that is associated (perhaps implicitly) with a prior function $\s(\cdot)$ is typically obtained by an optimization problem similar to $\mathrm{Prox}_{\sigma^2 \s}(x)$.

The proximal operator has been originally explored 
under the assumption that $\s$ is a 
lower semi-continuous (l.s.c.) convex function~\cite{moreau1965proximite}, for which it has been shown \cite[Prop.~1]{gribonval2020characterization} that: 1) it is a subdifferential of a convex function and 2) it is $1$-Lipschitz, i.e., nonexpansive. These properties are important for guaranteeing the convergence of algorithms that use it.
A recent work \cite{gribonval2020characterization} has extended the analysis to nonconvex $\s$ (which may be more similar to practical denoisers), but these results do not imply nonexpansiveness without additional assumptions on the proximal operator. 

Throughout this section, we make the following assumptions.

\begin{assump}
\label{ass:rho}
$\rho$ (the distribution parameters) is fixed.
\end{assump}

\begin{assump}
\label{ass:denoiser}
$\mathcal{D}$ is a proximal operator associated with some proper l.s.c. convex function.
\end{assump}

Note that the data-terms that we consider (namely, the log-likelihood and the MoM objectives) are challenging nonconvex functions.
They remain nonconvex even if $\rho$, which parameterizes the signal's distribution, is known or fixed. Similarly, the optimization problems remain nonconvex even when making the above simplifying assumption on the denoiser.
Nevertheless, note that our approach can be readily applied with denoisers that are based on TV regularization or on $\ell_1$-norm prior, which satisfy \ref{ass:denoiser}.

The two iterative estimation schemes in this paper can be expressed as $x_{k+1}=\mathcal{D}(\mathcal{T}(x_k))$, where the denoiser $\mathcal{D}$ follows an operation $\mathcal{T}$ that reduces some nonconvex 
data-fidelity term: 
the first scheme reduces the MoM objective with a gradient-based method and the second scheme reduces the negative log-likelihood via the EM method. 
We now turn to analyze two such schemes.

\subsection{Convergence of projected MoM}

Under \ref{ass:rho} and \ref{ass:denoiser}, 
let us provide convergence results for the iterative scheme 
\begin{align}
\label{eq:pgm}
    x_{k+1} = \mathcal{D} (x_k - \mu \nabla f (x_k)),
\end{align}
where $f$ denotes the MoM objective 
and $\mu$ is the step-size.
Notice that $f$ is smooth but nonconvex.
Compared to Algorithm~\ref{proj_MoM}, the main difference is merely replacing the BFGS steps with gradient steps.

To establish a convergence statement for \eqref{eq:pgm}, we will exploit convergence results of the proximal gradient method \cite{beck2009gradient,beck2017first}. 
To invoke them, we show that the MoM has a Lipschitz gradient.
The following theorem shows that this is the case when the optimization is restricted to a closed domain~$\mathcal{X}$. This assumption is reasonable when $x$ is an image, as typically its pixels are restricted to the range $[0, 255]$.

\begin{theorem}
Assume that $\mathcal{X}$ is a closed convex set. 
Let $\rho$ satisfy \ref{ass:rho} and $\mathcal{D}$ satisfy \ref{ass:denoiser} with respect to a function $\mu \s(\cdot)$, with $\mu\in \mathbb{R}$ and domain $\mathcal{X}$.
Consider the nonconvex MoM objective 
\begin{eqnarray}
&& \hspace{-0.5in} f(x)=\|PC_{{x}}D_{\rho}C_{{x}}^TP^T + \sigma^2PP^T -\hat{M_2}\|_{\F}^2  \\ \nonumber && +\lambda\|PC_{\rho}{x}-\hat{M_1}\|_{2}^2.
\end{eqnarray}
Then, there exists $\mu>0$ such that the sequence $\{ x_k \}$ in~\eqref{eq:pgm} obeys:
(1) the sequence $\| x_k-\mathcal{D} (x_k - \mu \nabla f (x_k)) \|_2$ converges to zero, and (2) any limit of $\{ x_k \}$
is a stationary point of $f+\s$.
\end{theorem}

\noindent \emph{Proof.}
We need to show that, restricted to $\mathcal{X}$, $\nabla f$ is $C$-Lipschitz. Then by setting $\mu \in (0,\frac{2}{C})$ 
statements (1) and (2) in the theorem follow
from \cite[Theorem 10.15]{beck2017first}. 

Note that $f(x)$ is infinitely differentiable.
The Lagrange remainder (a consequence of the mean value theorem) of first-order Taylor's series of $\nabla f$ implies that for any $x,z \in \mathcal{X}$ there exists $\xi$ on the line that connects $x$ and $z$ for which $\nabla f(x) = \nabla f(z) + \nabla^2 f(\xi) (x - z)$. Therefore,
$$
\|\nabla f(x) - \nabla f(z)\|_2 \leq \max_{\xi \in \mathcal{X}} \|\nabla^2 f(\xi) \|_{op} \| x - z \|_2,
$$
where $\|\cdot\|_{op}$ is the operator norm. As $f$ is a polynomial of order 4, the entries of $\nabla^2 f$ are polynomials of order 2 and are bounded on the closed set $\mathcal{X}$. Thus, the operator norm (magnitude of the largest eigenvalue of $\nabla^2 f$) is bounded on $\mathcal{X}$ due to the Gershgorin circle theorem. Denoting this finite value by $C$, we get that $\nabla f(x)$ is $C$-Lipschitz on $\mathcal{X}$.
\qed

Note that our statement
ensures that the limit is only
a stationary point of the 
objective $f+\s$, and not to a global minimum. This is inevitable, as $f$ is a nonconvex function.
Extensions to backtracking line search, instead of fixed step size, can be adopted from~\cite{beck2009gradient,beck2017first}.

\subsection{Composition of convergence inducing operators}

Unlike projected MoM, in projected EM the data-fidelity step is not a gradient step.
It is well known that, on its own, the EM algorithm $x_{k+1} = \mathcal{T}(x_k)$ converges to a local minimum of the negative log-likelihood function.
However, to the best of our knowledge, EM convergence behavior is not well understood, besides a few simple models. 
This motivates us to consider a framework of a composition of two ``black-box" operators that generate convergent sequences when being applied alone. 
Indeed, 
it is also known that the sequence $x_{k+1} = \mathcal{D}(x_k)$, 
under \ref{ass:denoiser},
converges to a minimizer of its associated function $\s$ (as a special case of the proximal gradient method on $0+\s$~\cite{beck2009gradient,beck2017first}).
Furthermore, a gradient-step operator with a sufficiently small step size on the MoM objective also yields a convergent sequence.
Therefore, the analysis we provide for EM may also apply to the MoM.

Denote by $\mathrm{Fix}(\mathcal{M})$ the set of fixed points of an operator $\mathcal{M}$, i.e., $\mathrm{Fix}(\mathcal{M}) = \{ x: \mathcal{M}(x)=x \}$. 
As discussed above, both the data-fidelity operator $\mathcal{T}$ and the denoising operator $\mathcal{D}$ have fixed points associated with stationary points of a data-fidelity term $f$ and an implicit prior term $\s$, respectively.

Let us now define a class of operators that strictly attract points towards their set of fixed points, as done in~\cite{yamada2005hybrid}.

\begin{definition}
We say that the map $\mathcal{M}:\mathcal{X} \to \mathcal{X}$ is 
$\gamma$-strongly quasi-nonexpansive in $\mathcal{X}$ with $\gamma>0$
if for all $x \in \mathcal{X}$ and $z \in \mathrm{Fix}(\mathcal{M})$ 
we have 
$$
\|\mathcal{M}(x) - z\|_2^2 \leq \| x-z\|_2^2 - \gamma \|\mathcal{M}(x) - x \|_2^2.
$$
\end{definition}

Note that this class includes the set of proximal operators (see \cite[Fact~1c]{yamada2005hybrid}). Thus, it is a weaker assumption.
Even though it is not guaranteed that the data-fidelity operator $\mathcal{T}$ that is used in EM falls into this class for the entire domain, it might be reasonable to restrict the domain, where it holds to be a small set $\mathcal{X}$ around a fixed point. 
The following theorem shows how this property implies the convergence of the sequence $x_{k+1} = \mathcal{D} (\mathcal{T} (x_k))$ to a fixed point.

\begin{theorem}
Let $\mathcal{D}$ satisfy \ref{ass:denoiser} with domain $\mathcal{X}$
and let $\mathcal{T}$ be $\gamma$-strongly quasi-nonexpansive in $\mathcal{X}$ such that $\mathrm{Fix}(\mathcal{D}) \cap \mathrm{Fix}(\mathcal{T}) \neq \emptyset$.
Then, the sequence
$x_{k+1} = \mathcal{D} (\mathcal{T} (x_k))$
converges to a fixed point $x_* = \mathcal{D} (\mathcal{T}(x_*))$. 
\end{theorem}

\noindent \emph{Proof.}
As a proximal operator associated with a convex function, $\mathcal{D}$ is $1/2$-average of a nonexpansive operator \cite[Theorem 6.42]{beck2017first},\cite[Lemma 2.1]{combettes2004solving}. Thus, from \cite[Fact~1c]{yamada2005hybrid}, it is strongly quasi-nonexpansive. From \cite[Prop.~1d]{yamada2005hybrid} the composition of strongly quasi-nonexpansive operators  $\mathcal{M}=\mathcal{D} \circ \mathcal{T}$
with $\mathrm{Fix}(\mathcal{D}) \cap \mathrm{Fix}(\mathcal{T})$ 
is strongly quasi-nonexpansive in $\mathcal{X}$ with some $\tilde{\gamma}>0$
and $\mathrm{Fix}(\mathcal{M}) = \mathrm{Fix}(\mathcal{D}) \cap \mathrm{Fix}(\mathcal{T})$.
Let $\tilde{x}_* \in \mathrm{Fix}(\mathcal{M})$, then 
\begin{eqnarray*}
&& \|x_{k+1} - \tilde{x}_*\|_2^2 
\leq \| x_k - \tilde{x}_*\|_2^2 - \gamma \|x_{k+1} - x_k \|_2^2 \\    
& \Rightarrow & \|x_{k+1} - x_k \|_2^2 \leq \frac{1}{\gamma} (\| x_k - \tilde{x}_*\|_2^2 - \|x_{k+1} - \tilde{x}_*\|_2^2).
\end{eqnarray*}
Summing over $k$ from $0$ to $K$ we have
\begin{align*}
\sum_{k=0}^{K} \|x_{k+1} - x_k \|_2^2 &\leq \frac{1}{\gamma} ( \| x_0 - \tilde{x}_*\|_2^2 - \|x_{K+1} - \tilde{x}_*\|_2^2 ) \\ \nonumber
&\leq \frac{1}{\gamma} \| x_0 - \tilde{x}_*\|_2^2.
\end{align*}
By taking $K$ to infinity, we see that $\{ x_k \}$ is a Cauchy sequence, and thus it converges to some limit point $x_*$, i.e., $x_k \rightarrow x_*$. Consequently, this limit point is a fixed-point:
$x_* = \mathcal{D} (\mathcal{T}(x_*))$.
\qed

The above result, despite including assumptions that are not guaranteed for EM on the entire domain, provides mathematical reasoning for the convergence of composition of a denoiser and a data-fitting operator, when each obeys convergence properties alone and both can agree on a fixed point (potentially, out of many fixed points that fit the measurements but can badly estimate the true signal). Note that the theorem assumption $\mathrm{Fix}(\mathcal{D}) \cap \mathrm{Fix}(\mathcal{T}) \neq \emptyset$ is reasonable, as we empirically show that projected EM converges to a noiseless point 
that agrees with the measurements.

\begin{figure}[t]
	\centering
	\includegraphics[width=\linewidth]{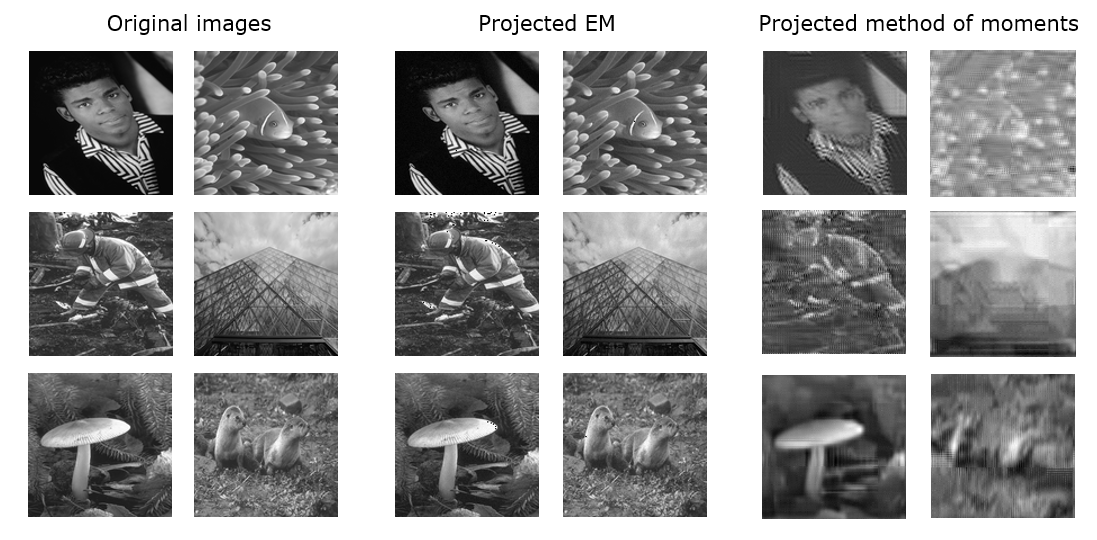}
	\caption{
		A gallery of 6 examples from the CBSD-68 dataset, and the corresponding recoveries using projected EM in a high-SNR regime ($\sigma=1/8$) and projected MoM with accurate moments ($N\gg \sigma^4$); the images were down-sampled by a factor of $K=2$. It can be seen that projected EM outperforms projected MoM. However, as we show in Table~\ref{tbl:MoMvsPMoM}, the computational load of projected EM is much heavier.
	}
	\label{fig:cbsd_images}
\end{figure}

\section{Numerical experiments}
\label{sec:experiments}

We applied the proposed algorithms
to the 68 ``natural" images of the CBSD-68 dataset and a cryo-EM image of the E. coli 70S ribosome, 
available at the Electron Microscopy Data Bank\footnote{\url{https://www.ebi.ac.uk/emdb/}}.
We set $\Lh=128$, and the images were down-sampled to $\Ll=64,32,16,8$ (namely, down-sampling factors of $K=2,4,6,8,$ respectively).

For all experiments, we used the BM3D denoiser~\cite{dabov2007image}, where the noise level decays as in~\eqref{eq:decaying_function}. 
Unless specified otherwise, we set $F=5$ for both algorithms. 
For the MoM, we minimized the least squares objective~\eqref{eq:LS} using the BFGS algorithm with line-search~\cite{bendory2022multi,kreymer2022two}.
The distributions $\rho_1, \rho_2$ and their initial guesses were drawn from a uniform distribution on $[0,1]$, and normalized so that $\rho_1, \rho_2 \in \bigtriangleup^{\Lh}$. 
The pixel values of all images are in $[0,1]$.
Each pixel in the initialization of the image estimate, for both algorithms,  was drawn i.i.d.\ from a uniform distribution on $[0,1]$,  which was then normalized to include the entire $[0,1]$ range.
The algorithms were initialized with a single initial guess; considering more initializations did not make a significant effect on the results.

Due to the shift-invariance of~\eqref{eq:model_1}, we measure the error by:
\begin{equation}
    \text{error}(\hx) = \min_s
    \frac{||R_{s}\hx-x||_\F}{||x||_\F},
\end{equation}
where $\hx$ is the estimated image, and $R_s$ denotes a 2D translation. 
We define SNR as 
\begin{equation}
	\text{SNR}= \frac{||x||_\text{F}^2}{\Lh^2\sigma^2}.
\end{equation}

Our algorithms were implemented in Torch.
The recovery of a high-dimensional signal requires an efficient implementation of the algorithms over GPUs.
Actually, without this implementation, we were unable to recover images of size $\Lh\geq32$ within a reasonable running time, where reducing the resolution beyond that adversely affects denoising performance. The code to reproduce all experiments is available at \url{https://github.com/JonathanShani/Denoiser_projection}.

\begin{table}
	\caption{
	Comparing projected MoM and EM for various target resolutions $\Lh$ of the Lena image, observed resolutions $\Ll$, observations number $N$ and $\sigma$. 
	}
	\label{tbl:MoMvsPMoM}
	\begin{center}
		\begin{tabular}{||c c c c || c c || c c||} 
		    \hline
		    \multicolumn{4}{||c||}{Model parameters}&\multicolumn{2}{|c||}{Projected MoM}&\multicolumn{2}{|c||}{Projected EM} \\
			\hline
			$\Lh$ & $\Ll$ & $N$ & $\sigma$ & Error & Time & Error & Time  \\

			\hline\hline
			
			$128$ & $32$ & $100$ & $0.125$ & $0.332$ & $\mathbf{266.5421}$ & $\mathbf{0.0991}$ & $ 871.6759 $ \\ 
			\hline 
			$128$ & $32$ & $100$ & $0.25$ & $0.3459$ & $ \mathbf{207.2537} $ & $\mathbf{0.196}$ & $ 861.4668 $ \\ 
			\hline 
			$128$ & $32$ & $100$ & $0.5$ & $\mathbf{0.356}$ & $\mathbf{120.2959}$ & $0.3836$ & $ 837.9343 $ \\ 
			\hline 
			$128$ & $32$ & $10000$ & $0.125$ & $0.339$ & $\mathbf{261.8394} $ & $\mathbf{0.0311}$ & $ 1593.9534 $ \\ 
			\hline 
			$128$ & $32$ & $10000$ & $0.25$ & $0.3302$ & $ \mathbf{280.6154} $ & $\mathbf{0.0623}$ & $ 1668.756 $ \\ 
			\hline 
			$128$ & $32$ & $10000$ & $0.5$ & $0.3301$ & $ \mathbf{277.3122} $ & $\mathbf{0.1387}$ & $ 1660.7012 $ \\ 
			\hline 
			$128$ & $16$ & $100$ & $0.125$ & $0.3627$ & $ \mathbf{85.716} $ & $\mathbf{0.2475}$ & $ 919.5462 $ \\ 
			\hline 
			$128$ & $16$ & $100$ & $0.25$ & $\mathbf{0.411}$ & $ \mathbf{84.4553} $ & $0.4154$ & $ 984.6871 $ \\ 
			\hline 
			$128$ & $16$ & $100$ & $0.5$ & $\mathbf{0.3714}$ & $ \mathbf{77.8746} $ & $0.7497$ & $ 893.2808 $ \\ 
			\hline 
			$128$ & $16$ & $10000$ & $0.125$ & $0.3874$ & $ \mathbf{122.5611} $ & $\mathbf{0.0685}$ & $ 3769.9855 $ \\ 
			\hline 
			$128$ & $16$ & $10000$ & $0.25$ & $0.381$ & $ \mathbf{113.7814} $ & $\mathbf{0.141}$ & $ 3179.0298 $ \\ 
			\hline 
			$128$ & $16$ & $10000$ & $0.5$ & $0.3652$ & $ \mathbf{94.8278} $ & $\mathbf{0.3001}$ & $ 3148.3869 $ \\ 
			\hline 
			$64$ & $32$ & $100$ & $0.125$ & $0.3319$ & $ 103.8799 $ & $\mathbf{0.0494}$ & $ \mathbf{100.9239} $ \\ 
			\hline 
			$64$ & $32$ & $100$ & $0.25$ & $0.3427$ & $ \mathbf{47.959} $ & $\mathbf{0.0988}$ & $ 99.215 $ \\ 
			\hline 
			$64$ & $32$ & $100$ & $0.5$ & $0.3439$ & $ \mathbf{44.0597} $ & $\mathbf{0.2111}$ & $ 87.1056 $ \\ 
			\hline 
			$64$ & $32$ & $10000$ & $0.125$ & $0.2562$ & $ \mathbf{101.4311} $ & $\mathbf{0.0158}$ & $ 325.7816 $ \\
			\hline 
			$64$ & $32$ & $10000$ & $0.25$ & $0.3382$ & $ \mathbf{100.1693} $ & $\mathbf{0.0315}$ & $ 264.8022 $ \\ 
			\hline 
			$64$ & $32$ & $10000$ & $0.5$ & $0.3349$ & $ \mathbf{98.7658} $ & $\mathbf{0.0647}$ & $ 276.1571 $  \\

			\hline
		\end{tabular}
	\end{center}
	\vspace{-0.2in}
\end{table}

\noindent {\bf Visual examples.}
Figure~\ref{fig:cbsd_images} presents a collection of visual recoveries from the CBSD-68 dataset using both algorithms, with a  down-sampling factor of $K=2$. 
For projected MoM, we assumed to have access to the population (exact) moments 
(which is the case when $N\gg\sigma^4$), and projected EM used $N=10^4$ observations with a noise level of $\sigma=1/8$ (corresponding to SNR$\approx$15). Evidently, projected EM provides more accurate recoveries, although MoM uses exact moments.  
However, as we show next, the computational load of projected EM is much heavier. 
Figure~\ref{fig:proj_em_visual} presents recoveries of additional images using projected EM with sampling factors of $K=2,4,8$, $N=10^4$ observations and noise level of $\sigma=1/8$. 
Remarkably, projected EM provides accurate estimates with $K=8$, i.e., when the number of pixels is reduced from $128^2$ to $16^2$.

\noindent {\bf Quantitative comparisons.}
We compared projected EM and projected MoM in terms of estimation error and run-time.
While the reported running times depend on implementation and hardware, 
the goal is to show the general trend.

Table~\ref{tbl:MoMvsPMoM} presents the error and running time, averaged over 20 trials, for each combination of the parameters $\Lh$, $\Ll$, $N$ and $\sigma$.
We used the Lena image of size $\Lh=128$ with down-sampling factors of $K=4,8$, and of size $\Lh=64$ with $K=2$; the number of observations was set to $N=10^2,10^4$, and the noise level to $\sigma=0.125,0.25,0.5$ (corresponding to SNR$\approx$17, 4, 1, respectively).
Both algorithms were limited to $100$ iterations per trial.
For each trial, we drew a fresh set of observations based on a new distribution and initial guesses.
Projected EM shows a clear advantage for almost all values of $N$ and $\sigma$ in terms of estimation error (besides some cases when $\sigma$ is large and $N$ is small).
However, the computational burden of EM is much heavier; there are cases (e.g., $\Lh=128, N=10^4,\sigma=0.25$) where the running time of projected EM algorithm is $30$ times longer than the running time of projected MoM. 
The reason is that EM iterates over all observations, whereas MoM requires only a single pass over the observations~\cite{boumal2018heterogeneous}, 
while the dimension of the variables in the least squares objective~\eqref{eq:LS_SR} is proportional to the dimension of the image.

\noindent {\bf Error as a function of the iterations.}
We next examined the behavior of the algorithms as a function of their iterations, averaged over all images in the CBSD-68 dataset of size $\Lh=128$. 
Projected MoM used exact moments, and projected EM used $N=10^4$ observations and a noise level of $\sigma=1/8$.
Figure~\ref{fig:convergence_graphs} presents results for a variety of sampling factors.
As can be seen, the algorithms with projections significantly outperform the unprojected versions for all down-sampling factors. Notably, the projected algorithms provide non-trivial estimates even for $K=16$ (when each measurement is down-sampled from $128\times 128$ 
to $8\times 8$).  
We note that the denoiser leads to a locally non-monotonic error.

\begin{figure}
	\centering
	\includegraphics[width=0.5\textwidth]{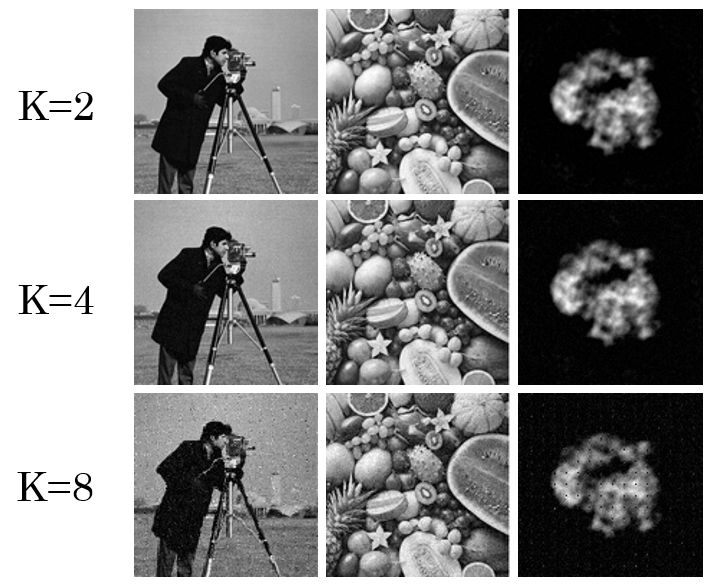}
	\caption{Recoveries using the projected EM algorithm with noise level of $\sigma=1/8$ and $N=10^4$ observations. The down-sampling factors are $K=2$ (top row), $K=4$ (middle row), and $K=8$ (bottom row). 
	}
	\label{fig:proj_em_visual}
\end{figure}

\begin{figure}[t]
    \centering
    \includegraphics[width=0.3\textwidth]{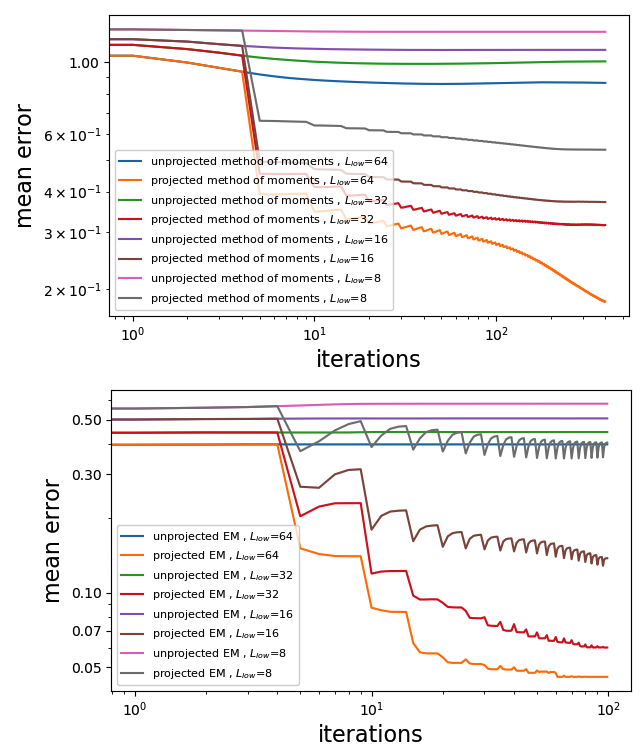}
    \caption{
    Error as a function of the iteration (averaged over all images in the CBSD-68 dataset, each of size $\Lh=128$), for different sampling factors. Top: projected MoM and the unprojected MoM, based on the perfect moments (corresponding to $N\to\infty$). Bottom:
    projected and unprojected EM, with $N=10^4$ observations and noise level of $\sigma=1/8$.
    }
    \label{fig:convergence_graphs}
\end{figure}

\noindent {\bf Projected MoM performance as a function of SNR.}
\label{sec:SNR_experiment}
Figure~\ref{fig:SNR_vs_error} (left panel) shows the average recovery error of projected MoM as a function of the SNR.  These experiments require averaging over multiple trials, which makes it virtually impossible to produce results for EM in a reasonable running time.
All experiments were conducted with the Lena image of size $\Lh=32$, with a sampling factor of $K=2$, and $N=10^5$ observations. 
For each SNR value, we conducted 100 trials, each with a fresh distribution of translations.
As can be seen, the error decreases monotonically with the SNR. 
Remarkably, the slope of the curve increases as the SNR decreases, hinting that the sample complexity of the problem (the number of observations required to achieve a desired accuracy) increases in the low SNR regime; this is a known phenomenon in the MRA literature for models with no projection \cite{bendory2017bispectrum,abbe2018multireference,perry2019sample}.

\begin{figure}[t]
    \centering
    \includegraphics[width=0.24\textwidth]{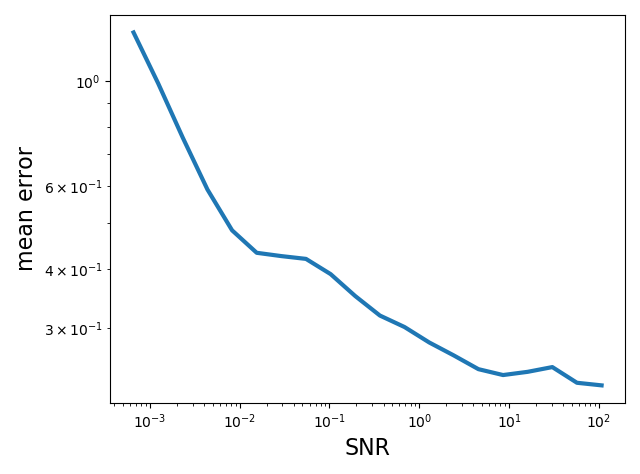}
    \includegraphics[width=0.24\textwidth]{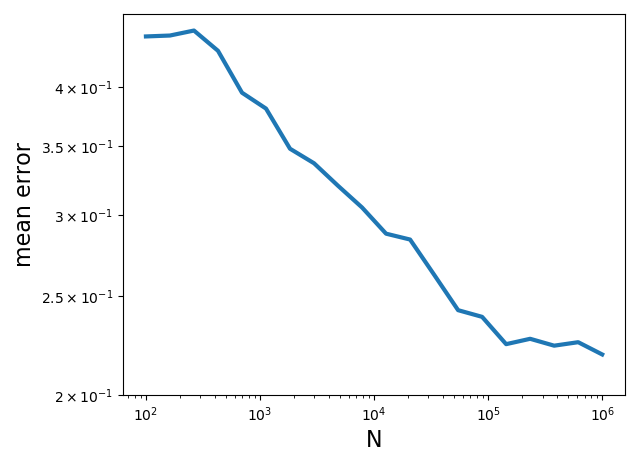}
    \caption{
    Mean error of projected MoM as a function of SNR with $N=10^5$ observations (left) and number of observations with $\sigma=1/8$ (right). The image size is $\Lh=32$  and the down-sampling factor is $K=2$.
    }
    \label{fig:SNR_vs_error}
\end{figure}

\noindent {\bf Projected MoM performance as a function of the number of measurements.}
By the law of large numbers and according to~\eqref{eq:moments_convergence}, for any fixed SNR, the empirical moments almost surely converge to the analytic moments. 
To assess the impact of the number of observations on the estimation error, we used the image of Lena of size $\Lh=32$, down-sampling factor of $K=2$, and noise level of $1/8$, corresponding to $\text{SNR}\approx{17}$.

Figure~\ref{fig:SNR_vs_error} (right panel) presents the average error, over 100 trials, as a function of $N$. For each trial, we drew a fresh distribution.
The error indeed decreases with $N$, as expected, but the slope is smaller than $-1/2$ (in logarithmic scale), as the law of large numbers predicts. 
This is due to the effect of the prior (manifested as a denoiser in our case), which does not depend on the observations (the data).

\noindent {\bf How frequently should we apply the denoiser?}
An important parameter of the proposed algorithms, denoted by $F$, determines after how many gradient steps (or EM iterations for projected EM) we apply the denoiser. 
For example, $F=1$ means that we apply the denoiser after each gradient step, whereas $F=100$ means that we run 100 gradient steps before applying the denoiser. 
This is especially important since the computational complexity of a single gradient step is significantly heavier than a BM3D application. 
To study the effect of this parameter for projected MoM, we used the image of Lena of size $\Lh=128$, a down-sampling factor of $K=2$, and assumed to have access to the population (perfect) moments. 

Figure~\ref{fig:F_influence} shows the average error (over 10 trials) of MoM for different values of $F$; each trial was conducted with a fresh distribution.
Plainly, the impact of $F$ on the performance is significant. 
The optimal value seems to be around $F=10$ in this case;  larger and smaller values of  $F$ lead to sub-optimal results.

\begin{figure}[t]
    \centering
    \includegraphics[width=0.3\textwidth]{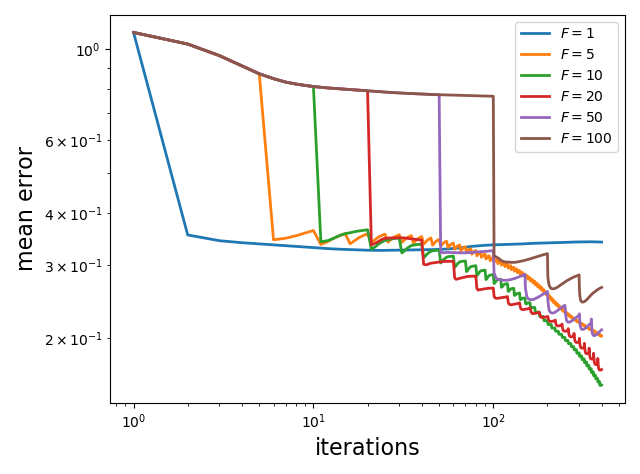}
    \caption{
    Mean error of projected MoM as a function of the iteration, for different values of $F$ (number of gradient steps between consecutive denoising operations).  The experiments were conducted on an image of size $\Lh=128$, assuming to have access to the population moments, and with a down-sampling factor of $K=2$.
The optimal value seems to be around $F=10$;  larger or smaller values of  $F$ lead to sub-optimal results.
    }
    \label{fig:F_influence}
\end{figure}

\section{Conclusion}
\label{sec:conclusion}
Single-particle cryo-EM is an increasingly popular technology for high-resolution structure determination of biological molecules~\cite{bendory2020single,Zhao2023Computational}.
The cryo-EM reconstruction problem  is to estimate a 3-D structure (the electrostatic  potential map of the molecule of interest) from multiple observations of the form
$y = T R_\omega X + \e$, 
 where $X$ is the 3-D structure to be recovered, $R_\omega$ represents an unknown 3-D rotation by some angle $\omega\in SO(3)$, and $T$ is a fixed, linear tomographic projection. The noise level is typically very high. 
The SR-MRA problem~\eqref{eq:model_1} can be interpreted as a toy model of the cryo-EM problem, where the image plays the role of the 3-D structure, and the 2-D translations and the sampling operator $P$ replace, respectively, the unknown 3-D rotations and the tomographic projection.

The standard reconstruction algorithms in cryo-EM are based on EM~\cite{scheres2012relion,punjani2017cryosparc}, aiming to maximize the posterior distribution based on analytical, explicit priors. 
Recently, techniques based on MoM were also designed for quick ab initio modeling~\cite{levin20183d,sharon2020method}. 
The techniques proposed in this paper can be readily integrated into these schemes, which may lead to improved accuracy and acceleration.
A similar idea was recently suggested by~\cite{kimanius2021exploiting}, who used
the regularization by denoising technique~\cite{RED}, and showed recoveries of moderate resolution with simulated data.

Current cryo-EM technology cannot be used to reconstruct small molecular structures (below $\sim 40$ kDa): this is one of its main drawbacks~\cite{henderson1995potential}.
Recent papers suggested overcoming this barrier using 
 variations of the MoM and EM ~\cite{bendory2018toward,kreymer2023stochastic,zabatani2024score}.
However, as the problem is ill-conditioned, the resolution of the recovered structures is not high enough.   
We hope that modifications of the algorithms suggested in this paper, where the denoiser is replaced with a data-driven
projection operator can be used to increase the resolution of the reconstructions and thus will allow elucidating multiple new structures of small biological molecules.

\rev{In our proposed scheme, the data fidelity (likelihood) term is computationally expensive due to the large number of measurements ($N > 100$) involved. A possible approach to effectively reduce the runtime is using online strategies, such as the one used in the online Plug-and-Play \cite{Sun2019Online}.}
Note also that in this work we use BM3D as prior in our proposed framework. \rev{Yet, one may replace BM3D with a learned denoiser to further boost the performance, e.g., using a deep neural network. We defer this extension to future work.}

\section*{Acknowledgment}
	T.B. is partially supported by the BSF grant no. 2020159, 
	NSF-BSF grant no. 2019752, and ISF grant no. 1924/21. R.G. is partially supported by the ERC-StG (No. 757497) and KLA grants. 
 T.T. is partially supported by ISF grant no. 1940/23.

\bibliographystyle{plain}

\end{document}